\documentclass[copyright,creativecommons]{eptcs}
\providecommand{\firstpage}{1}
\usepackage{graphics}
\usepackage{epsfig}
\usepackage{color}
\usepackage{amsmath}
\usepackage{amssymb}
\usepackage{textcomp}
\usepackage{xcolor}
\usepackage{listings}
\usepackage{theorem}
\usepackage{hyperref}

\usepackage{tikz}
\usetikzlibrary{decorations.pathmorphing} 
\usetikzlibrary{fit}                    
\usetikzlibrary{backgrounds}    
\usetikzlibrary{shapes,arrows}
\usepackage[version=3]{mhchem}
\usepackage{booktabs}

\newcommand{\spec}[1]{\texttt{#1}}

\renewcommand{\~}[0]{\texttildelow}

\newtheorem{notation}{Notation}[section]
\newtheorem{defn}{Definition}[section]
\newtheorem{thm}{Theorem}[section]
\newtheorem{exmp}{Example}[section]
\newenvironment{proof}[1][Proof]{\begin{trivlist}
 \item[\hskip \labelsep {\bfseries #1}]}{\end{trivlist}}

\def\authorrunning{L. Brim et al.}
\def\titlerunning{Compact Representation of Photosynthesis}

\title{Compact Representation of Photosynthesis Dynamics by Rule-based Models (Full Version)}

\author{L. Brim, J. Ni\v{z}nan, D. \v{S}afr\'anek
\thanks{The work has been supported by the EC OP project No. CZ.1.07/2.3.00/20.0256.}
\institute{Faculty of Informatics, Masaryk University\\
Brno, Czech Republic
}
\email{safranek@fi.muni.cz}
}

\begin{document}
\maketitle

\begin{abstract}

Traditional mathematical models of photosynthesis are based on mass action kinetics of light reactions. This approach requires the modeller to enumerate all the possible state combinations of the modelled chemical species. This leads to combinatorial explosion in the number of reactions although the structure of the model could be expressed more compactly. We explore the use of rule-based modelling, in particular, a simplified variant of Kappa, to compactly capture and automatically reduce existing mathematical models of photosynthesis. Finally, the reduction procedure is implemented in BioNetGen language and demonstrated on several ODE models of photosynthesis processes. 
This is an extended version of the paper published in proceedings of 5th International Workshop on Static Analysis and Systems Biology (SASB) 2014. 
\end{abstract}




\section{Introduction}

Photosynthesis is one of the most important biophysical processes driving life on Earth. Most life forms, including humans, depend on photosynthesis that transforms energy of solar radiation into energy-rich organic matter, releases oxygen that we breathe, and removes excess carbon dioxide from the atmosphere that would threaten the Earth's energy balance. Adding to the relevance of photosynthesis, significant expectations emerged lately in connection with potential human interventions in the global carbon cycle -- among the considered alternatives are the higher generation biofuels~\cite{BrennanOwende2010} or biomineralization by point-source carbon capture~\cite{JanssonNorthen2010}.

Current coarse-grained mathematical models of photosynthesis~\cite{ephotobook} cover the known parts of the entire process. They build up the light reactions dynamics from simplified interactions on and inbetween complicated protein complexes involved in the transfer of the energy from light into the cell. Many different local modifications at these protein structures are traversed after reception of the photon. To capture the process mechanistically, many elementary chemical reactions connect together to form the model. Each effective structure combination has to be enumerated in order to assign the appropriate kinetic laws. This inevitable expansion then leads to combinatorial explosion in the number of possible complexes. 


In~\cite{ephotosynthesis} we have developed an online repository for mathematical models of photosynthesis. That effort has opened many questions regarding the differing levels of available models and the problem of their formal representation  in a single suitably expressive formalism. After several years of interactions with modellers targeting photosynthesis, we now attempt to move towards practical applications of state-of-the-art formal methods in that field.  


Rule-based modelling~\cite{kappa_formal,innovations,rbm2014} is an approach that has been developed to tackle primarily the complexity of cell signalling systems where combinatorial explosion comes from configurations of phosphate bindings to specific sites of a protein. In particular, it allows us to compactly represent complicated models that would be tedious to specify using traditional reaction-based methods~\cite{agile,cyanoclock}.
The interactions between proteins are represented using rules at the level of functional components. In photosynthesis, there occurs a number of specific protein complex modifications that are in abstract essence similar to phosphorylation though crucially different at the side of physics.
The two most well-known representatives of rule-based modelling languages are Kappa~\cite{kappa_formal} and BioNetGen Language~\cite{bngl}.

On the theoretical side, we contribute to the increasing set of algebraic-based modelling efforts by employing a simplified version of Kappa calculus to compact a set of domain-specific models coming from biophysics of photosynthesis. We do not construct the models from the scratch but we rather take several existing kinetic models of photosynthesis-related processes and reformulate them in the algebraic framework. At the level of compositional representation, we formulate syntactic reductions of the models that preserve behavioural equivalence. 

On the practical side, we employ BioNetGen language (BNGL) and related tools~\cite{bngl,BNG} to implement the models. We contribute by automatising our reduction procedure in Python. By applying reductions to the considered models we show the power of process-algebraic framework to compactly represent combinatorially exploding systems of light reactions. 

The paper shows the importance of process-algebraic description for the domain of photosynthesis. To the best of our knowledge, the application of such techniques in the field of systems biology of photosynthesis is still at its beginning and we believe our contribution is useful to help to establish rule-based modelling in the domain.

\subsection{Related Work}

There are many applications of rule-based modelling available (see~\cite{rbm2014} for an overview). However, applications to photosynthesis are very rare.
In~\cite{rulebasedphotosynth}, the authors provide a model of chlorophyll $a$ fluorescence induction kinetics that is simulated in the rule-based framework by means of Monte Carlo simulation. The work is unique in the sense it pioneers rule-based approach for photosynthesis models. The contribution brings a stochastic simulation algorithm that reflects the differing context where the context-free rules appear. This is important, since the quantitative rates of electron transfers are modulated by aggregation of several modifications of the protein complex, i.e., the photosystem. The considered model is comparable to Laz\'ar model~\cite{lazar2009approaches} we use for our case study.


In~\cite{model_reductions}, the problem of combinatorial complexity in models where the quantitative semantics cannot be generated is addressed.
Several frameworks for abstracting the models at the level of semantics have been developed for Kappa.

\section{Background}

We define simplified Kappa
using a process-like notation as is presented in~\cite{danos2008abstract}, syntax and the notions of structural equivalence and matching are entirely take from~\cite{danos2008abstract}:

\begin{tabular}{ l l l l l}
    expression & $E : := \emptyset ~|~ a, E$            &~~~~~~~~& site              & $s : := n^{\lambda}_{\iota}$ \\
    agent      & $a : := N(\sigma)$                     &   & site name         & $n : := x \in \mathcal{S}$   \\
    agent name & $N : := A \in \mathcal{A}$             &   & internal state    & $\iota : := \epsilon ~|~ m \in \mathbb{V}$\\
    interface  & $\sigma : := \emptyset ~|~ s, \sigma$  &   & binding state     & $\lambda : := \epsilon ~|~ i \in \mathbb{N}$ \\
\end{tabular}

where $\mathcal{A}$ is a finite set of agent names, $\mathcal{S}$ is a finite set of site names, 
$\mathbb{V}$~is a finite set of values representing modified states of the sites. 
An agent is denoted by its name and its interface.
Interface consists of a sequence of sites.
$x_{\iota}^{\lambda}$ denotes a site $x$ with internal state $\iota$ and binding state $\lambda$.
If the binding state is $\epsilon$ then the site is free, otherwise it is bound.
By convention, when a binding or internal site is not specified, $\epsilon$ is considered.

Note that full Kappa is richer.
It allows a binding state meaning a \emph{free or bound site}, denoted by a question mark.
We also omit rates from the rules. 

\begin{defn}
An expression is \emph{well-formed} if a site name occurs only once in an interface and if 
each binding state ($\neq\epsilon$) present in the expression occurs exactly twice.
The \emph{set of all well-formed expressions} is denoted as $\mathcal{E}$.
The set of all well-formed expressions that can be generated from the literal $a$ is called the \emph{set of all well-formed agents} and is denoted as $\mathcal{E}_a$.
Similarly, $\mathcal{E}_{\sigma}$ denotes the \emph{set of all well-formed interfaces} and
$\mathcal{E}_{s}$ the \emph{set of all well-formed sites}.
\end{defn}

Next, we define some notations we use throughout the text.
\begin{notation}
    \emph{expressions} $E,~E'~\in~\mathcal{E}$; 
\emph{agents} $a,~a'~\in~\mathcal{E}_a$; \emph{agent name} $A~\in~\mathcal{A}$; 
\emph{interfaces} $\sigma,~\sigma'~\in~\mathcal{E}_\sigma$; sites $s,~s'~\in~\mathcal{E}_s$; \emph{site name} $x~\in~\mathcal{S}$;
\emph{internal state} $\iota~\in~\{\epsilon\}~\cup~\mathbb{V}$; \emph{specific internal state} $m~\in~\mathbb{V}$;
\emph{binding state} $\lambda~\in~\{\epsilon\}~\cup~\mathbb{N}$;
\end{notation}

Next, we provide inductive definitions of some useful mappings.

\begin{defn}
    \emph{Agent name} is a mapping $\mathsf{name}: \mathcal{E}_a \rightarrow \mathcal{A}$ defined as\linebreak $\mathsf{name}(A(\sigma)) = A$.

    We define \emph{agent sites} as a mapping $\mathsf{sites}: \mathcal{E}_a \rightarrow 2^\mathcal{S}$ such that
    $\mathsf{sites}(A()) = \emptyset$, 
    $\mathsf{sites}(A(x_{\iota}^{\lambda})) = \{x\}$, and
    $\mathsf{sites}(A(s, \sigma)) = \mathsf{sites}(A(s)) \cup \mathsf{sites}(A(\sigma))$.

    \emph{Agent internal state} is a mapping $\mathsf{state}: \mathcal{E}_a \rightarrow (\mathbb{V} \cup \{\epsilon\})^\mathcal{S}$ defined as
    $\mathsf{state}(A()) = \emptyset$, 
    $\mathsf{state}(A(x_{\iota}^{\lambda})) = \{(x, \iota)\}$, and
    $\mathsf{state}(A(s, \sigma)) = \mathsf{state}(A(s)) \cup \mathsf{state}(A(\sigma))$.
\end{defn}

\begin{defn} 
    \emph{Structural equivalence} $\equiv \subseteq \mathcal{E} \times \mathcal{E}$ is defined as
    a relation satisfying the following properties:
    \begin{enumerate}
        \item Reflexivity:\hfill
            $ E \equiv E $
        \item The order of sites in interfaces does not matter:\\\hspace*{1mm}\hfill
            $E, A(\sigma, s, s', \sigma'), E' \equiv E, A(\sigma, s', s, \sigma'), E' $
        \item The order of agents in an expression does not matter:\\\hspace*{1mm}\hfill
            $ E, a, a', E' \equiv E, a', a, E'$
        \item Binding states can be injectively renamed:\hfill
            ${E[i/j] \equiv E}$\\\hspace*{1mm}\hfill where $i, j \in \mathbb{N}$ and $i$ does not occur in $E$.
    \end{enumerate}
    \emph{Solution} $[E] \in 2^{\mathcal{E}}$ denotes the equivalence class of $E$ in $\equiv$.
    $\mathcal{L}$ is a set of all solutions.
\end{defn}

\begin{defn} 
A \emph{rule} is a pair of expressions $E_l$, $E_r$ (usually written as $E_l \rightarrow E_r$). 
The set of all rules is denoted as $\mathcal{R}$.
\end{defn}
The left hand side $E_l$ of the rule describes the solution taking part in the reaction and
the right hand side $E_r$ describes the effects of the rule.
The rule can be either a binding rule or a modification rule.
A binding (unbinding) rule binds two free sites together (or unbinds two bound sites).
A modification rule modifies some internal state~\cite{danos2008abstract}.

\setlength{\tabcolsep}{12pt}
\begin{defn} 
    \label{matching-def}
    \emph{Matching} is a relation denoted as $\models \subseteq \mathcal{E} \times \mathcal{E}$ and defined inductively in the left column below.
    \emph{Replacement} is a function $\mathcal{E} \times \mathcal{E}\rightarrow\mathcal{E}$ defined in the right column. 
    \begin{center}
\begin{tabular}{ c c }
    $n_{\iota}^{\lambda} \models n_{\iota}^{\lambda}$ & 
        $n_{\iota}^{\lambda}[n_{\iota_r}^{\lambda_r}] = n_{\iota_r}^{\lambda_r}$  \\[0.15cm]
    $n_{\iota}^{\lambda} \models n^{\lambda}$ & 
        $n_{\iota}^{\lambda}[n^{\lambda_r}] = n_{\iota}^{\lambda_r}$  \\[0.15cm]
    $\sigma \models \emptyset$ & $\sigma[\emptyset] = \sigma$ \\[0.15cm]
    $\frac{\displaystyle s \models s_l \quad \sigma \models \sigma_l}{\displaystyle s,\sigma \models s_l,\sigma_l}$ &
        $s,\sigma[s_r,\sigma_r] = s[s_r], \sigma[\sigma_r]$ \\ [0.3cm]
    $\frac{\displaystyle \sigma \models \sigma_l}{\displaystyle N(\sigma) \models N(\sigma_l)}$ &
        $N(\sigma)[N(\sigma_r)] = N(\sigma[\sigma_r])$ \\ [0.3cm]
    $E \models \emptyset$ & $E[\emptyset] = E$ \\ [0.15cm]
    $\frac{\displaystyle a \models a_l \quad E \models E_l}{\displaystyle a,E \models a_l,E_l}$ &
        $(a, E)[a_r,E_r] = a[a_r], E[E_r]$ \\
\end{tabular}
    \end{center}
\end{defn}

A replacement can be applied only if the corresponding matching is satisfied.

In order to apply a rule $E_l \rightarrow E_r$ to a solution $[E]$ the expression $E$ representing the solution must 
first be reordered to an equivalent expression $E'$ that matches $E_l$ (according to the definition
of matching stated above).
$E'$~is then replaced with $E'[E_r]$ (also defined above).

\begin{defn}
    \emph{Rule application} is a mapping $\tau: \mathcal{L} \times \mathcal{R} \to \mathcal{L}$ such that
    $ \tau([E], (E_l, E_r)) = [E'[E_r]] \text{ whenever }$ $\exists E' \in [E]. E' \models E_l$.
\end{defn}
Rules yield a \emph{transition system} between solutions containing an edge $[E] \rightarrow_{E_l, E_r} [E'[E_r]]$ 
whenever $\exists E' \in [E]. E' \models E_l$.

\begin{defn} 
    An \emph{agent signature} $(\Sigma, I)$ is a pair of mappings $\Sigma: \mathcal{A} \rightarrow 2^{\mathcal{S}}$ and 
    $I: \mathcal{A} \times \mathcal{S} \rightarrow 2^{\mathbb{V}}$.
\end{defn}
Informally, $\Sigma$ restricts for each agent name $A \in \mathcal{A}$ the set of site names that can occur in an agent with name $A$. 
And $I$ restricts the set of internal states a particular site can attain.

\begin{defn}
    $E$ satisfies agent signature $(\Sigma,I)$, denoted $(\Sigma,I)\vdash E$, if $E$~satisfies one of the following conditions:
    \begin{enumerate}
        \item $E \equiv \emptyset$
        \item $E \equiv A()$ and $A \in dom(\Sigma)$
        \item $E \equiv A(x_{\epsilon}^{\lambda})$ and $x \in \Sigma(A)$
        \item $E \equiv A(x_m^{\lambda})$ and $x \in \Sigma(A)$ and $m \in I(A, x)$
        \item $E \equiv A(s, \sigma)$ where $(\Sigma, I) \vdash A(s)$ and $(\Sigma, I) \vdash A(\sigma)$
        \item $E \equiv E_l, A(\sigma)$ where $(\Sigma, I) \vdash E_l$ and $(\Sigma, I) \vdash A(\sigma)$ 
    \end{enumerate}

    If $r = (E_l, E_r) \in \mathcal{R}$ and $(\Sigma, I) \vdash E_l$ and $(\Sigma, I) \vdash E_r$ then $(\Sigma, I) \vdash r$.

    If $R \subseteq \mathcal{R}$ and $\forall r \in R. (\Sigma, I) \vdash r$ then $(\Sigma, I) \vdash R$.
\end{defn}

\begin{defn}
    An agent $a$ is \emph{complete} with respect to signature $(\Sigma,I)$, denoted $(\Sigma,I)\models a$, if 
    $\mathsf{sites}(a)=\Sigma(\mathsf{name}(a)) \land \forall x\in\mathsf{sites}(a). \mathsf{state}(a)(x)\in I(\mathsf{name}(a),x)$.

    An expression $E$ is \emph{complete} with respect to signature $(\Sigma, I)$, denoted $(\Sigma, I) \models E$, if it satisfies one of the following conditions:
    \begin{enumerate}
        \item $E \equiv \emptyset$
        \item $E \equiv a, E'$ where $a \in \mathcal{E}_a, E' \in \mathcal{E}$ and $(\Sigma, I) \models a$ and $(\Sigma, I) \models E'$ 
    \end{enumerate}

    $\mathcal{E}_{(\Sigma, I)} = \{E \in \mathcal{E} | (\Sigma, I) \models E \}$ is a set of all expressions that are complete with respect to signature $(\Sigma, I)$ .
\end{defn}

\begin{defn}
    A \emph{rule-based model} $\mathcal{M}$ is a tuple $(\Sigma,I,R)$ that satisfies the condition $(\Sigma,I)\vdash R$.
    We use the notation $\mathsf{Signature}(\mathcal{M}) = (\Sigma, I)$, 
        $\mathsf{Rules}(\mathcal{M}) = R$,
        $\mathcal{M} \vdash E \iff (\Sigma, I) \vdash E$ for $E \in \mathcal{E}$,
        $\mathcal{M} \models E \iff (\Sigma, I) \models E$ for $E \in \mathcal{E}$, and
        $\mathcal{E}_{\mathcal{M}} = \mathcal{E}_{(\Sigma, I)}$.
\end{defn}

\begin{defn}
    An \emph{initialised model} $M$ is a pair $(\mathcal{M}, E_i)$ where $\mathcal{M}$ is a rule-based model and
    $E_i$ is an expression representing the \emph{initial solution} such that $\mathcal{M} \models E_i$.
\end{defn}

\begin{defn}
    \label{rn-def}
    A \emph{state space} of an initialised model $M = (\mathcal{M}, E_i)$ is a pair $(\mathsf{Solutions}(M) \subseteq \mathcal{L}, \mathsf{Reactions}(M)$ $\subseteq \mathcal{L} \times \mathcal{L})$ defined inductively as follows:
    \begin{enumerate}
        \item $[E_i] \in \mathsf{Solutions}(M)$
        \item $[E] \in \mathsf{Solutions}(M)$ and $\exists r \in \mathsf{Rules}(\mathcal{M}). \tau([E], r) = [E']$ \\ if and only if
            $[E'] \in \mathsf{Solutions}(M)$ and $([E], [E']) \in \mathsf{Reactions}(M)$
    \end{enumerate}
\end{defn}

\begin{defn}
    Initialised models $M_1 = (\mathcal{M}_1, E_1)$ and $M_2 = (\mathcal{M}_2, E_2)$ are \emph{structurally equivalent}, 
    denoted $M_1 \equiv M_2$,
    if and only if $\mathsf{Solutions}(M_1) = \mathsf{Solutions}(M_2)$ and $\mathsf{Reactions}(M_1) = \mathsf{Reactions}(M_2)$.
\end{defn}

\begin{defn}
    Models $\mathcal{M}_1$ and $\mathcal{M}_2$ are \emph{structurally equivalent}, 
    denoted $\mathcal{M}_1 \equiv \mathcal{M}_2$,
    if and only if $\forall E_i \in \mathcal{E}_{\mathcal{M}_1} \cup \mathcal{E}_{\mathcal{M}_2}.(\mathcal{M}_1, E_i) \equiv (\mathcal{M}_2, E_i) $.
\end{defn}

In BNGL, agents are called molecules and they are specified in a similar manner as in the simplified Kappa.
An example of a molecule is \spec{A(x\~n!1)} where the site \spec{x} has an internal state \spec{n} (separated from the site by a tilde) and a binding state is \spec{1} (separated by the exclamation mark). 
The BNGL alternatives to agent signatures are called molecule types they are defined using the notation demonstrated in the following example: \spec{A(x\~n\~b, y\~n\~a)}.
Here, the allowed internal states of the individual sites are separated by tildes (site \spec{x} can have an internal state \spec{n} or \spec{b}).
Rules are described by the $lhs$ \spec{->} $rhs$ notation (or $lhs$ \spec{<->} $rhs$ in the case of reversible rules).
The individual model components (molecule types, reaction rules, seed species, observables) are in BNGL separated by the \spec{begin} \emph{keyword} and \spec{end} \emph{keyword} pairs.

\section{Model Reductions}

In this section, we formally define several syntactic operations that can be used to reduce rule-based models. In particular, we assume an original ODE model to be directly transferred to a rule-based model. At that level we apply syntactic reductions to remove redundancies of the original model. 
As a motivation the following example can be considered.
 \begin{exmp}
     \label{exmp-cee}
     Assume a comprehensive model of photosystem II protein complex containing the following rules:
 \[ P(Yz_0, P680_+, ChlD_0, Pheo_0) \rightarrow P(Yz_+, P680_0, ChlD_0, Pheo_0) \]
 \[ P(Yz_0, P680_+, ChlD_0, Pheo_-) \rightarrow P(Yz_+, P680_0, ChlD_0, Pheo_-) \]

 It is obvious that $Pheo$ is not affected by these rules in any way. 
 Since it has only two possible states ($0$ and $+$),
 the two rules can be reduced into the following single rule where $Pheo$ has an empty internal state:
 \[ P(Yz_0, P680_+, ChlD_0, Pheo) \rightarrow P(Yz_+, P680_0, ChlD_0, Pheo) \]
\end{exmp}

To capture the syntactic manipulation mentioned above, we introduce an operation called \emph{context enumeration elimination} of $Pheo$~in~$P$.
Note that it looks like we could eliminate $Pheo$ totally, but notice that it is unbound (i.e. its binding state is $\epsilon$).

\begin{defn}
    \label{cee-def}
    Model $\mathcal{M}_1$ and model $\mathcal{M}_2$ are in relation \emph{context enumeration elimination},
    ($(\mathcal{M}_1, \mathcal{M}_2) \in \rho_{cee}$)
, iff
    $\mathsf{Signature}(\mathcal{M}_1) = \mathsf{Signature}(\mathcal{M}_2) = (\Sigma, I)$
    and $\exists~A~\in~\mathcal{A}, x~\in~\mathcal{S}, \lambda \in \{\epsilon\} \cup \mathbb{N}, E_l,~E_r~\in~\mathcal{E},$ $\sigma_l,~\sigma_r~\in~\mathcal{E}_{\sigma}$ such that
\begin{enumerate}
\item  $ \mathsf{Rules}(\mathcal{M}_1) \setminus \mathsf{Rules}(\mathcal{M}_2) = \{(E_l^m, E_r^m) | m \in I(A,x)\}$
    where $\forall m \in I(A,x):\,$ $E_l^m \equiv E_l, A(\sigma_l, x^{\lambda}_m)\text{ and } E_r^m \equiv E_r, A(\sigma_r, x^{\lambda}_m)$,
\item $ \mathsf{Rules}(\mathcal{M}_2) \setminus \mathsf{Rules}(\mathcal{M}_1) = \{(E_l', E_r')\}$ where
    $ E_l' \equiv E_l, A(\sigma_l, x^{\lambda}_{\epsilon})\text{ and }
    E_r' \equiv E_r, A(\sigma_r, x^{\lambda}_{\epsilon})$.
\end{enumerate}
\end{defn}

\begin{thm}
    Context enumeration elimination preserves structural equivalence of models.
    If $(\mathcal{M}_1, \mathcal{M}_2) \in \rho_{cee}$ then $\mathcal{M}_1 \equiv \mathcal{M}_2$.
\end{thm}
\begin{proof}
    Let $(\mathcal{M}_1, \mathcal{M}_2) \in \rho_{cee}$. Then $\mathcal{E}_{\mathcal{M}_1} = \mathcal{E}_{\mathcal{M}_2}$.
    Let $E_i \in \mathcal{E}_{\mathcal{M}_1}$, $M_1 = (\mathcal{M}_1, E_i)$ and $M_2 = (\mathcal{M}_2, E_i)$.
    We prove that $M_1 \equiv M_2$ 
    by induction through the structure of their state spaces.
    Without a loss of generality we can fix the variables used in Definition~\ref{cee-def}.
    \begin{enumerate}
        \item  From Definition~\ref{rn-def}:
            $[E_i] \in \mathsf{Solutions}(M_1)$, $[E_i] \in \mathsf{Solutions}(M_2) $
        \item \emph{Completeness}:\\
            Let $[E] \in \mathsf{Solutions}(M_1)$ and $r \in \mathsf{Rules}(\mathcal{M}_1). \tau([E], r) = [E']$.\\
            From induction we have $[E] \in \mathsf{Solutions}(M_2)$.
            \begin{enumerate}
                \item $r \in \mathsf{Rules}(\mathcal{M}_1) \cap \mathsf{Rules}(\mathcal{M}_2)$.
                    Then we have $r \in \mathsf{Rules}(\mathcal{M}_2)$.
                \item $r \in \mathsf{Rules}(\mathcal{M}_1) \setminus \mathsf{Rules}(\mathcal{M}_2)$.
                    So $\exists m \in I(A, x). r = (E_l^m, E_r^m)$. \\
                    Let $r' = (E_l', E_r') \in \mathsf{Rules}(\mathcal{M}_2) \setminus \mathsf{Rules}(\mathcal{M}_1)$
                    and $e \in [E]. e \models E_l^m$. Then from Definition~\ref{matching-def} we have
                    $e \models E_l'$ and $e[E_r^m] = e[E_r']$. Therefore, $\tau([E], r') = [E']$.
            \end{enumerate}
            Thus  $\mathsf{Solutions}(M_1) \subseteq \mathsf{Solutions}(M_2)$,
            $\mathsf{Reactions}(M_1) \subseteq \mathsf{Reactions}(M_2)$.
            
        \item \emph{Soundness}:\\
            Let $[E] \in \mathsf{Solutions}(M_2)$ and $r \in \mathsf{Rules}(\mathcal{M}_2). \tau([E], r) = [E']$. \\
            From induction we have $[E] \in \mathsf{Solutions}(M_1)$.
            \begin{enumerate}
                \item $r \in \mathsf{Rules}(\mathcal{M}_2) \cap \mathsf{Rules}(\mathcal{M}_1)$.
                    Then we have $r \in \mathsf{Rules}(\mathcal{M}_1)$.
                \item $r \in \mathsf{Rules}(\mathcal{M}_2) \setminus \mathsf{Rules}(\mathcal{M}_1)$.
                    So $r = (E_l', E_r')$. Let $e \in [E]. e \models E_l'$.
                    Then there must be $x_m^\lambda$ in $e$ that gets matched to the $x_\epsilon^\lambda$ part $E_l'$.
                    It must be that $m \in I(A, x)$ and so $r' = (E_l^m, E_r^m) \in \mathsf{Rules}(\mathcal{M}_1) \setminus \mathsf{Rules}(\mathcal{M}_2)$.
                    From Definition~\ref{matching-def} we have $e \models E_l^m$ and $e[E_r^m] = e[E_r']$.
                    Therefore, $\tau([E], r') = [E']$.
            \end{enumerate}
            Thus  $\mathsf{Solutions}(M_2) \subseteq \mathsf{Solutions}(M_1)$,
            $\mathsf{Reactions}(M_2) \subseteq \mathsf{Reactions}(M_1)$.
    \end{enumerate}
\end{proof}

\begin{exmp}
     In Example \ref{exmp-cee} we did not remove $Pheo$ from
     \[ P(Yz_0, P680_+, ChlD_0, Pheo) \rightarrow P(Yz_+, P680_0, ChlD_0, Pheo) \]
     because it would make this rule applicable in solutions where it was not applicable before the change.
     The new rule
     \[ P(Yz_0, P680_+, ChlD_0) \rightarrow P(Yz_+, P680_0, ChlD_0) \]
     could be applied in some solutions where $Pheo$ is bound.
     Therefore, this operation does not preserve the structural semantics of every model. 
     However, if in an initialized model, $Pheo$ is not initialized in a bound state and there are no rules that would change it into a bound state
    then the result of this operation is equivalent with the original initialized model.
\end{exmp}

\begin{defn}
    Model $\mathcal{M}_1$ is in relation \emph{generic unbound context elimination} 
    with model $\mathcal{M}_2$, 
    $(\mathcal{M}_1, \mathcal{M}_2) \in \rho_{guce}$, iff
    $\mathsf{Signature}(\mathcal{M}_1) = \mathsf{Signature}(\mathcal{M}_2)$
    and $\exists A \in \mathcal{A}, x \in \mathcal{S}, E_l, E_r \in \mathcal{E}, \sigma_l, \sigma_r \in \mathcal{E}_{\sigma}$ such that
\begin{enumerate}
\item $\mathsf{Rules}(\mathcal{M}_1) \setminus \mathsf{Rules}(\mathcal{M}_2) = \{(E_l^1, E_r^1)\}$
    where $E_l^1 \equiv E_l, A(\sigma_l, x^\epsilon_\epsilon)$ and 
    $E_r^1 \equiv E_r, A(\sigma_r, x^\epsilon_\epsilon)$,
\item $ \mathsf{Rules}(\mathcal{M}_2) \setminus \mathsf{Rules}(\mathcal{M}_1) = \{(E_l^2, E_r^2)\}$
    where
    $E_l^2 \equiv E_l, A(\sigma_l)$ and $E_r^2 \equiv E_r, A(\sigma_r)$.
\end{enumerate}
\end{defn}

This reduction is useful in models where the removed contexts are guaranteed not to be bound. 

\begin{exmp}
     Models of photosynthesis can contain rules such as the following \emph{antenna deexcitation rule}:
     \[ P(P680_n, Qa_n,PhD_n,ac_*,ChlD_n) \rightarrow P(P680_n,Qa_n,PhD_n,ac_n,ChlD_n)\]
     We can see that this rule contains a lot of specific contexts -- sites with specified internal states whose internal state is not changed in the rule ($P680_n, Qa_n, PhD_n, ChlD_n$).
     We want to remove these contexts if they do not change the semantics of the model.
     We can define an operation that eliminates a single context from a rule.
     After the operation is applied to an initialized model we can check if the original and the resulting models have the same state space.
     \end{exmp}

Motivated by the previous example, we formally define the operation of removing a specific context from a model. 

\begin{defn}
    Model $\mathcal{M}_1$ is in relation \emph{specific unbound context elimination} 
    with model $\mathcal{M}_2$, 
    $(\mathcal{M}_1, \mathcal{M}_2) \in \rho_{suce}$, iff
    $\mathsf{Signature}(\mathcal{M}_1) = \mathsf{Signature}(\mathcal{M}_2)$
    and $\exists A \in \mathcal{A}, x \in \mathcal{S}, E_l, E_r \in \mathcal{E}, \sigma_l, \sigma_r \in \mathcal{E}_{\sigma}, m \in \mathbb{V}$ such that
\begin{enumerate}
\item $ \mathsf{Rules}(\mathcal{M}_1) \setminus \mathsf{Rules}(\mathcal{M}_2) = \{(E_l^1, E_r^1)\}$
    where
    $E_l^1 \equiv E_l, A(\sigma_l, x^\epsilon_m)$ and
    $E_r^1 \equiv E_r, A(\sigma_r, x^\epsilon_m)$,
\item $\mathsf{Rules}(\mathcal{M}_2) \setminus \mathsf{Rules}(\mathcal{M}_1) = \{(E_l^2, E_r^2)\}$
    where $E_l^2 \equiv E_l, A(\sigma_l)$ and $E_r^2 \equiv E_r, A(\sigma_r)$.
\end{enumerate}
\end{defn}

This reduction can be used if the set of reachable solutions in an intialized model to which this rule can be applied is not affected by the reduction.

Sometimes it may be useful to eliminate a rule from a model. 
Reasons for doing so can be different. 
For example, one might want to see how the behaviour of a model changes after the rule is removed.
Or if the rule is not reachable in an initialised model then it can be removed to reduce the size of the model description.
\begin{defn} 
    Model $\mathcal{M}_1$ is in relation \emph{rule elimination} 
    with model $\mathcal{M}_2$,
    denoted $(\mathcal{M}_1, \mathcal{M}_2) \in \rho_{re}$, if and only if
    $\mathsf{Signature}(\mathcal{M}_1) = \mathsf{Signature}(\mathcal{M}_2)$
    and $\exists r \in \mathcal{R}$ such that
    \scalebox{.95}{$\mathsf{Rules}(\mathcal{M}_1) \setminus \mathsf{Rules}(\mathcal{M}_2) = \{r\}$}
    and 
    \scalebox{.95}{$ \mathsf{Rules}(\mathcal{M}_2) \setminus \mathsf{Rules}(\mathcal{M}_1) = \emptyset$}.
\end{defn}

If a rule is not reachable in some initialised model then we can safely remove it without affecting the semantics of the initialised model.

\section{Application to Photosynthesis Models} 

In this section, we describe the application of the reductions to several models of photosynthesis.
If the model is reaction-based, we first rewrite it to rule-based form.

\subsection{Implementation}

We used the library PySB~\cite{pysb} for the specification of the photosynthetic models and to automatise their export to BNGL.
The scripts that implement syntactic operations have been written in Python.
We used BioNetGen for constructing and simulating the models. The scripts are available at
\url{https://github.com/jniznan/rbm-photosynthesis}.

We search the space of possible models that can be constructed by applying syntactic operations to the original model by depth-first search. We stop when we find a model that cannot be further reduced.
There can be multiple models that cannot be further reduced. Our algorithm finds only one.
We apply the syntactic operations in a given order: (i) context enumeration elimination, (ii) generic/specific context elimination, (iii) rule elimination. 
This approach is a heuristic that attempts to maximize the number of reductions.

\subsection{Photosynthesis}
Light-dependent reactions begin in photosystem II where the photons hit and excite the antenna molecules.
The excitation then travels via a chain of proteins until it arrives to chlorophyll \emph{a}.
Or a photon can directly excite chlorophyll \emph{a}.
This excitation causes the primary electron acceptor (pheophytin) to accept an electron from chlorophyll \emph{a} species called P680 -- the primary electron donor.
The electron is exchanged by multiple protein molecules until it reaches plastoquinone.
The electron missing from chlorophyll \emph{a} is replenished through a tyrosine residue from so called oxygen-evolving complex that strips electrons from water molecules, producing molecular oxygen and hydrogen protons into the lumen.

\begin{figure}[ht]
\begin{center}
\includegraphics[width=0.7\textwidth]{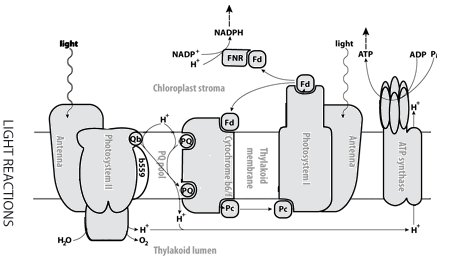}
\caption{Schema of the thylakoid membrane where light-dependent reactions occur.
}
\label{fig:thylakoid}
\end{center}
\end{figure}

After plastoquinone accepts two electrons, it is converted to its \ce{PQH2} form by accepting two hydrogen protons from the chloroplast stroma.
Then it travels to cytochrome \ce{b6f} where it is converted back to its original \ce{PQ} form, leaving the two electrons in the cytochrome and the two hydrogen protons in the lumen.
After that, plastoquinone returns to photosystem II, ready to accept other electrons.
As is shown in Figure~\ref{fig:thylakoid}, the electrons travel through plastocyanin, photosystem I, ferredoxin to ferredoxin-NADP reductase where they are used for converting \ce{NADP+} to \ce{NADPH}.
The process of the electrons travelling from the oxygen-evolving complex to the ferredoxin-NADP reductase is known as the \emph{Z-scheme} of light.
Hydrogen protons that are left in the lumen are pumped back into the chloroplast stroma by powering ATP synthase which uses that energy to convert ADP into ATP.


\subsection{Comprehensive Model of Photosystem II}
Photosynthesis is well-adapted for spatially and temporally heterogeneous environments~\cite{nedbal2007}.
The light conditions under which photosynthesis occurs are constantly changing for many reasons,
including day-night cycle, moving clouds or changing seasons.
Mechanisms lying behind the adaptability of photosynthesis are not yet fully understood~\cite{rascher2006}.
Evidence confirms that photosynthesis in fluctuating light is more dynamic than simply adapting to the light extremes.

Rules are used to informally capture the model, but no existing formal rule-based language is employed.
This fact makes the CMS a perfect candidate for rewriting it into an RBM format.
Since the reaction rates are not available the model is worked out at a qualitative level of view.


The original model contains 22 rules.
By applying several syntactic operations we are able to reduce the model size to just 17 rules with smaller contexts.

We have also considered a reduced variant of the model that concentrates on  light absorption by electron transfers inside PSII. That allowed the authors~\cite{holzwarth2006} to introduce a fully specified kinetic model. By employing the reductions we were able to compact the model rules but not to decrease their number. 

\subsubsection{Reduction Procedure}

In an idealized case, PSII consists of 9 principal components with several possible states: Peripheral antenna (ground, excited), Core antenna (ground, excited), Mn Cluster ($S_0$, $S_1$, $S_2$, $S_3$, $S_4$), $Y_z$ (neutral, oxidized), P680 (neutral, oxidized), ChlD (ground, excited, oxidized), Pheo (neutral, reduced), $Q_A$ (neutral, reduced), and $Q_B$ (neutral, reduced, 2-reduced, $PQH_2$, empty).

When modelling PSII using chemical reactions or ODEs, each state in the model is a combination of the states of its components.
This leads to 4800 possible combinatorial states.
In~\cite{nedbal2007}, the authors argue that this combinatorial complexity makes a highly comprehensive model, that uses all these components and states, practically impossible.
Such model could contain $4800 \times 4800 = 23.04 \times 10^6$ reactions. 
This complexity would get larger if this comprehensive model were to include other models of similar complexity.
Instead of trying to build this large model, the authors~\cite{nedbal2007} reduce the system to model the behaviour they are interested in.

The system is reduced to model the effect of exposing PSII to ultra-short saturating flash of light on the fluorescence emission~\cite{nedbal2007}. 
The model that the authors reduce is based on~\cite{holzwarth2000} and~\cite{holzwarth2006}.
They assume that the fluorescence detectors used in measurements are not fast enough to capture processes occurring during the ultra-short flash of light.
Therefore, they are explicitly using an initial state of the model with antenna already excited.
Due to the choice of the modelled behaviour they do not consider the peripheral components of PSII (the Mn cluster, $Y_z$ and $Q_B$).
Another type of reduction they demonstrate is the biochemical removal of the peripheral antenna, leaving only core antenna and chlorophyll donor to hold the excitation.

The reduced model contains only the following 5 components leading to 48 possible combinatorial states: Core antenna (ground, excited), P680 (neutral, oxidized), ChlD (ground, excited, oxidized), Pheo (neutral, reduced), and $Q_A$ (neutral, reduced).
The following is the rule-based BNGL representation of the model:\\
{\small
\begin{verbatim}
begin molecule types
  PSII(P680~n~p,Qa~n~m,PhD1~n~m,ac~n~exc,ChlD1~n~p~exc)
end molecule types

begin reaction rules
  #deexcitation_antenna_fluorescence:        
  PSII(P680~n,Qa~n,PhD1~n,ac~exc,ChlD1~n) -> PSII(P680~n,Qa~n,PhD1~n,ac~n,ChlD1~n)    
  #deexcitation_antenna_heat:                
  PSII(P680~n,Qa~n,PhD1~n,ac~exc,ChlD1~n) -> PSII(P680~n,Qa~n,PhD1~n,ac~n,ChlD1~n)   
  #deexcitation_ChlD1_fluorescence:          
  PSII(P680~n,Qa~n,PhD1~n,ac~n,ChlD1~exc) -> PSII(P680~n,Qa~n,PhD1~n,ac~n,ChlD1~n)   
  #deexcitation_ChlD1_heat:                  
  PSII(P680~n,Qa~n,PhD1~n,ac~n,ChlD1~exc) ->  PSII(P680~n,Qa~n,PhD1~n,ac~n,ChlD1~n)    
  #excitation_transfer:                      
  PSII(P680~n,Qa~n,PhD1~n,ac~exc,ChlD1~n) <-> PSII(P680~n,Qa~n,PhD1~n,ac~n,ChlD1~exc)   
  #primary_charge_separation_recombination:  
  PSII(P680~n,Qa~n,PhD1~n,ac~n,ChlD1~exc) <-> PSII(P680~n,Qa~n,PhD1~m,ac~n,ChlD1~p)  
  #stable_pair_generation_degeneration:   
  PSII(P680~n,Qa~n,PhD1~m,ac~n,ChlD1~p) <-> PSII(P680~p,Qa~n,PhD1~m,ac~n,ChlD1~n)    
  #quinone_qa_reduction_oxidation:           
  PSII(P680~p,Qa~n,PhD1~m,ac~n,ChlD1~n) <-> PSII(P680~p,Qa~m,PhD1~n,ac~n,ChlD1~n)   
end reaction rules
\end{verbatim}
}


Many of the contexts in the rules could be removed.
We therefore automatically apply a series of \emph{specific unbound context eliminations} to get a reduced model. The result is the following:
{\small \begin{verbatim}
begin reaction rules
  #deexcitation_antenna_fluorescence:        
  PSII(ac~exc) -> PSII(ac~n)  
  #deexcitation_antenna_heat:                
  PSII(ac~exc) -> PSII(ac~n)   
  #deexcitation_ChlD1_fluorescence:          
  PSII(ChlD1~exc) -> PSII(ChlD1~n)  
  #deexcitation_ChlD1_heat:                  
  PSII(ChlD1~exc) -> PSII(ChlD1~n)   
  #excitation_transfer:                      
  PSII(ac~exc,ChlD1~n) <-> PSII(ac~n,ChlD1~exc)   
  #primary_charge_separation_recombination:  
  PSII(PhD1~n,ChlD1~exc) <-> PSII(PhD1~m,ChlD1~p)   
  #stable_pair_generation_degeneration:      
  PSII(P680~n,PhD1~m,ChlD1~p) <-> PSII(P680~p,PhD1~m,ChlD1~n)    
  #quinone_qa_reduction_oxidation:       
  PSII(Qa~n,PhD1~m,ChlD1~n) <-> PSII(Qa~m,PhD1~n,ChlD1~n)    
end reaction rules
\end{verbatim}
}

\emph{Specific unbound context elimination} does not decrease the number of rules in the model.
However, we can see that the rules become much more compact.
The increased clarity of the rules makes it easy to see their effects.

In the \emph{stable pair generation/degeneration} rule, one may notice the \spec{PhD1\~m} specified in the context. 
Similarly, in the \emph{quinone Qa reduction/oxidation} rule, \spec{ChlD1\~n} is specified in the context.
Next, we take a look on how the behavior of the model changes if we try to remove one of these contexts.

\begin{figure}[ht]
\begin{center}
\includegraphics[width=\textwidth]{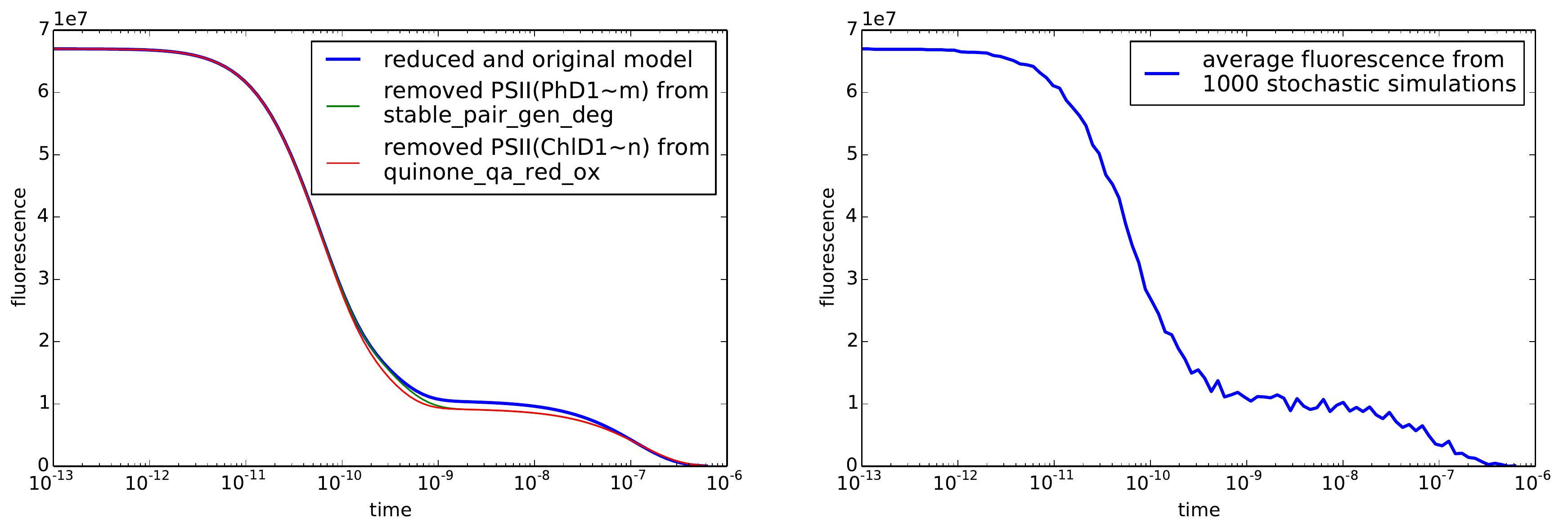}
\caption{ On the left: Fluorescence curve from the reduced model (blue) and the model where we manually remove also the context \spec{PhD1\~m} (green) or \spec{ChlD1\~m} (red). 
On the right: Fluorescence curve obtained by averaging 1000 stochastic simulations of the reduced model. The curves show fluorescence absorption time-course in relative units.
 }
\label{fig:holzwarth_rbm_sim}
\end{center}
\end{figure}

Figure~\ref{fig:holzwarth_rbm_sim} (left) compares fluorescence curves from several models.
The blue line shows fluorescence for the original model. 
This curve is the same for the reduced model we get after automatically reducing the model.
The green curve shows the changes that occur after we remove \spec{PhD1\~m} from the \emph{stable pair generation/degeneration} rule. We see that the fluorescence level is slightly lower on the plateau in the middle of the chart.
Similar situation occurs when we remove \spec{ChlD1\~n} from the context of \emph{quinone Qa reduction/oxidation}.
The resulting fluorescence levels are shown in red.

To illustrate why the fluorescence curves differ we compare the reaction networks of the reduced model and the reduced model with \spec{PhD1\~m} removed from context.
Figure~\ref{fig:holzwarth_rn} on the left shows the reaction network of the automatically reduced model. 
This reaction network is the same as the reaction network of the original model.
On the right, the reaction network of the additionally reduced model is shown.
We can see that this reaction network contains one more reachable agent that is not reachable in the original model.
This difference results in the different fluorescence curves as shown before.

\begin{figure}[ht]
\begin{center}
\includegraphics[width=\textwidth]{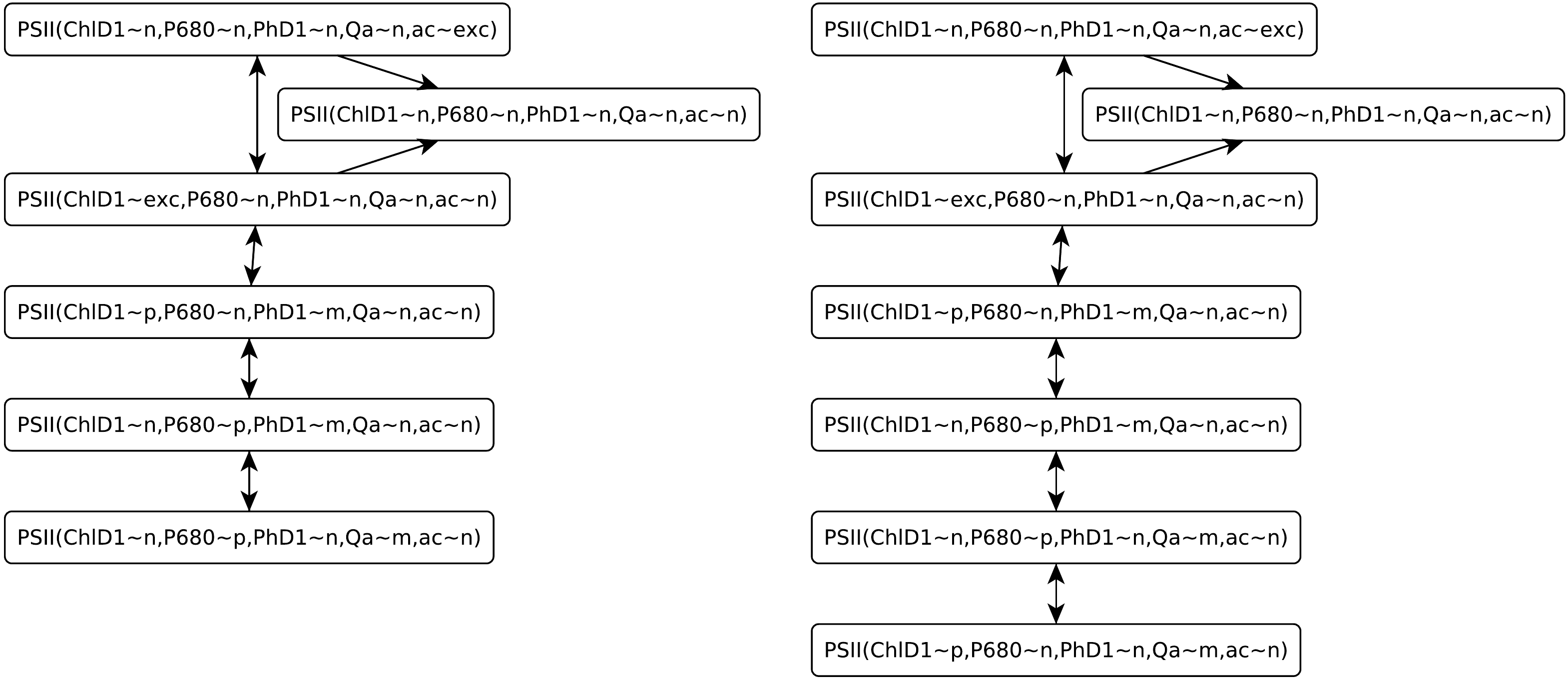}
\caption{ Reaction network of the reduced model (left) and of the reduced model with \spec{PhD1\~m} additionaly removed from the context (right).
 }
\label{fig:holzwarth_rn}
\end{center}
\end{figure}

The reduced model have the same reaction network as the original model for the initial condition \spec{PSII(P680~n,Qa~n,PhD1~n,ac~exc,ChlD1~n)}. 
For some other initial conditions, the original model would produce a reaction network consisting only of the initial solution, whereas the reduced model could produce a larger network.
This should not matter, since the authors~\cite{nedbal2007} make it explicit that the model should be used only with the specified initial condition.

\subsection{Integrated Model of Light-reactions}\label{lazar_cs}


Model Laz\'{a}r~\cite{lazar2009big} targets all important light-dependend thylakoid membrane processes participating in the \emph{Z-scheme} of light.
The model contains the following complexes and their parts:
photosystem II (with parts P680, Qa, Qb); oxygen evolving complex (with states $S_i$ where $i = 0, 1, 2, 3$); PQ,PQH - plastoquinone; cytochrome \ce{b6f} (with parts bL, bHc, f); Pc - plastocyanin; photosystem I (with parts P700, Fb); 
Fd - ferredoxin; ferredoxin-NADP reductase.
The model is expressed by 69 differential equations that can be browsed in detail at \url{http://www.e-photosynthesis.org}.
As we show later, this model has great potential for succint representation.


The rule-based reformulation of the full model in BNGL is the following:\\
{\small\begin{verbatim}
begin molecule types
  PSII(P680~n~p,Qa~n~m,Qb~n~m~2m)
  PQ()
  PQH()
  S(x~0~1~2~3)
  CytB6F(bL~n~m,bHc~n~m~2m,f~n~m)
  Fd(x~n~m)
  Pc(x~n~p)
  PSI(P700~n~p,Fb~n~m)
  FNR(x~i~a~am~a2m)
end molecule types

begin reaction rules
  #electron_transport_from_cytochrome_to_oxidized_Pc_1:  
  CytB6F(bL~n,bHc~n,f~m) + Pc(x~p) <-> CytB6F(bL~n,bHc~n,f~n) + Pc(x~n)
  #electron_transport_from_cytochrome_to_oxidized_Pc_2:  
  CytB6F(bL~m,bHc~n,f~m) + Pc(x~p) <-> CytB6F(bL~m,bHc~n,f~n) + Pc(x~n)
  #electron_transport_from_cytochrome_to_oxidized_Pc_3:  
  CytB6F(bL~n,bHc~m,f~m) + Pc(x~p) <-> CytB6F(bL~n,bHc~m,f~n) + Pc(x~n) 
  #electron_transport_from_cytochrome_to_oxidized_Pc_4:  
  CytB6F(bL~m,bHc~m,f~m) + Pc(x~p) <-> CytB6F(bL~m,bHc~m,f~n) + Pc(x~n)  
  #electron_transport_from_cytochrome_to_oxidized_Pc_5:  
  CytB6F(bL~n,bHc~2m,f~m) + Pc(x~p) <-> CytB6F(bL~n,bHc~2m,f~n) + Pc(x~n)
  #electron_transport_from_cytochrome_to_oxidized_Pc_6:  
  CytB6F(bL~m,bHc~2m,f~m) + Pc(x~p) <-> CytB6F(bL~m,bHc~2m,f~n) + Pc(x~n)
  #electron_transport_from_reduced_Fd_to_cytochrome_1:   
  Fd(x~m) + CytB6F(bL~n,bHc~n,f~n) <-> Fd(x~n) + CytB6F(bL~n,bHc~m,f~n) 
  #electron_transport_from_reduced_Fd_to_cytochrome_2:   
  Fd(x~m) + CytB6F(bL~m,bHc~n,f~n) <-> Fd(x~n) + CytB6F(bL~m,bHc~m,f~n)   
  #electron_transport_from_reduced_Fd_to_cytochrome_3:   
  Fd(x~m) + CytB6F(bL~n,bHc~n,f~m) <-> Fd(x~n) + CytB6F(bL~n,bHc~m,f~m)   
  #electron_transport_from_reduced_Fd_to_cytochrome_4:   
  Fd(x~m) + CytB6F(bL~m,bHc~n,f~m) <-> Fd(x~n) + CytB6F(bL~m,bHc~m,f~m)   
  #electron_transport_from_reduced_Fd_to_cytochrome_5:   
  Fd(x~m) + CytB6F(bL~n,bHc~m,f~n) <-> Fd(x~n) + CytB6F(bL~n,bHc~2m,f~n)  
  #electron_transport_from_reduced_Fd_to_cytochrome_6:
  Fd(x~m) + CytB6F(bL~m,bHc~m,f~n) <-> Fd(x~n) + CytB6F(bL~m,bHc~2m,f~n)  
  #electron_transport_from_reduced_Fd_to_cytochrome_7:
  Fd(x~m) + CytB6F(bL~n,bHc~m,f~m) <-> Fd(x~n) + CytB6F(bL~n,bHc~2m,f~m)  
  #electron_transport_from_reduced_Fd_to_cytochrome_8:
  Fd(x~m) + CytB6F(bL~m,bHc~m,f~m) <-> Fd(x~n) + CytB6F(bL~m,bHc~2m,f~m)  
  #electron_transport_inside_cytochrome_1:
  CytB6F(bL~m,bHc~n,f~n) <-> CytB6F(bL~n,bHc~m,f~n)   
  #electron_transport_inside_cytochrome_2:
  CytB6F(bL~m,bHc~n,f~m) <-> CytB6F(bL~n,bHc~m,f~m)   
  #electron_transport_inside_cytochrome_3:
  CytB6F(bL~m,bHc~m,f~n) <-> CytB6F(bL~n,bHc~2m,f~n)  
  #electron_transport_inside_cytochrome_4:
  CytB6F(bL~m,bHc~m,f~m) <-> CytB6F(bL~n,bHc~2m,f~m)  
  #electron_transport_from_reduced_PQ_to_cytochrome_1:
  PQH() + CytB6F(bL~n,bHc~n,f~n) <-> PQ() + CytB6F(bL~m,bHc~n,f~m)   
  #electron_transport_from_reduced_PQ_to_cytochrome_2:
  PQH() + CytB6F(bL~n,bHc~m,f~n) <-> PQ() + CytB6F(bL~m,bHc~m,f~m)    
  #electron_transport_from_reduced_PQ_to_cytochrome_3:
  PQH() + CytB6F(bL~n,bHc~2m,f~n) <-> PQ() + CytB6F(bL~m,bHc~2m,f~m)  
  #electron_transport_from_cytochrome_to_oxidized_PQ_1:
  PQ() + CytB6F(bL~n,bHc~2m,f~n) <-> PQH() + CytB6F(bL~n,bHc~n,f~n)   
  #electron_transport_from_cytochrome_to_oxidized_PQ_2:
  PQ() + CytB6F(bL~m,bHc~2m,f~n) <-> PQH() + CytB6F(bL~m,bHc~n,f~n)   
  #electron_transport_from_cytochrome_to_oxidized_PQ_4:
  PQ() + CytB6F(bL~m,bHc~2m,f~m) <-> PQH() + CytB6F(bL~m,bHc~n,f~m)   
  #electron_transport_from_cytochrome_to_oxidized_PQ_3:
  PQ() + CytB6F(bL~n,bHc~2m,f~m) <-> PQH() + CytB6F(bL~n,bHc~n,f~m)   
  #electron_transport_from_Fd_to_FNR_1:
  FNR(x~a) + Fd(x~m) <-> FNR(x~am) + Fd(x~n)   
  #electron_transport_from_Fd_to_FNR_2:
  FNR(x~am) + Fd(x~m) <-> FNR(x~a2m) + Fd(x~n) 
  #electron_transport_from_photosystemI_to_Fd_1:
  Fd(x~n) + PSI(P700~n,Fb~m) <-> Fd(x~m) + PSI(P700~n,Fb~n)  
  #electron_transport_from_photosystemI_to_Fd_2:
  Fd(x~n) + PSI(P700~p,Fb~m) <-> Fd(x~m) + PSI(P700~p,Fb~n)  
  #activation_of_FNR:
  FNR(x~i) -> FNR(x~a)   
  #turnover_of_FNR:
  FNR(x~a2m) -> FNR(x~a)    
  #electron_donation_1_1:
  S(x~0) + PSII(P680~p,Qa~n,Qb~n) -> S(x~1) + PSII(P680~n,Qa~n,Qb~n)    
  #electron_donation_2_1:
  S(x~0) + PSII(P680~p,Qa~m,Qb~n) -> S(x~1) + PSII(P680~n,Qa~m,Qb~n)    
  #electron_donation_3_1:
  S(x~0) + PSII(P680~p,Qa~n,Qb~m) -> S(x~1) + PSII(P680~n,Qa~n,Qb~m)    
  #electron_donation_4_1:
  S(x~0) + PSII(P680~p,Qa~m,Qb~m) -> S(x~1) + PSII(P680~n,Qa~m,Qb~m)    
  #electron_donation_5_1:
  S(x~0) + PSII(P680~p,Qa~n,Qb~2m) -> S(x~1) + PSII(P680~n,Qa~n,Qb~2m)  
  #electron_donation_6_1:
  S(x~0) + PSII(P680~p,Qa~m,Qb~2m) -> S(x~1) + PSII(P680~n,Qa~m,Qb~2m)  
  #electron_donation_1_2:
  S(x~1) + PSII(P680~p,Qa~n,Qb~n) -> S(x~2) + PSII(P680~n,Qa~n,Qb~n)    
  #electron_donation_2_2:
  S(x~1) + PSII(P680~p,Qa~m,Qb~n) -> S(x~2) + PSII(P680~n,Qa~m,Qb~n)    
  #electron_donation_3_2:
  S(x~1) + PSII(P680~p,Qa~n,Qb~m) -> S(x~2) + PSII(P680~n,Qa~n,Qb~m)    
  #electron_donation_4_2:
  S(x~1) + PSII(P680~p,Qa~m,Qb~m) -> S(x~2) + PSII(P680~n,Qa~m,Qb~m)    
  #electron_donation_5_2:
  S(x~1) + PSII(P680~p,Qa~n,Qb~2m) -> S(x~2) + PSII(P680~n,Qa~n,Qb~2m) 
  #electron_donation_6_2:
  S(x~1) + PSII(P680~p,Qa~m,Qb~2m) -> S(x~2) + PSII(P680~n,Qa~m,Qb~2m) 
  #electron_donation_1_3:
  S(x~2) + PSII(P680~p,Qa~n,Qb~n) -> S(x~3) + PSII(P680~n,Qa~n,Qb~n)    
  #electron_donation_2_3:
  S(x~2) + PSII(P680~p,Qa~m,Qb~n) -> S(x~3) + PSII(P680~n,Qa~m,Qb~n)    
  #electron_donation_3_3:
  S(x~2) + PSII(P680~p,Qa~n,Qb~m) -> S(x~3) + PSII(P680~n,Qa~n,Qb~m)    
  #electron_donation_4_3:
  S(x~2) + PSII(P680~p,Qa~m,Qb~m) -> S(x~3) + PSII(P680~n,Qa~m,Qb~m)    
  #electron_donation_5_3:
  S(x~2) + PSII(P680~p,Qa~n,Qb~2m) -> S(x~3) + PSII(P680~n,Qa~n,Qb~2m)  
  #electron_donation_6_3:
  S(x~2) + PSII(P680~p,Qa~m,Qb~2m) -> S(x~3) + PSII(P680~n,Qa~m,Qb~2m)  
  #electron_donation_1_0:
  S(x~3) + PSII(P680~p,Qa~n,Qb~n) -> S(x~0) + PSII(P680~n,Qa~n,Qb~n)    
  #electron_donation_2_0:
  S(x~3) + PSII(P680~p,Qa~m,Qb~n) -> S(x~0) + PSII(P680~n,Qa~m,Qb~n)    
  #electron_donation_3_0:
  S(x~3) + PSII(P680~p,Qa~n,Qb~m) -> S(x~0) + PSII(P680~n,Qa~n,Qb~m)    
  #electron_donation_4_0:
  S(x~3) + PSII(P680~p,Qa~m,Qb~m) -> S(x~0) + PSII(P680~n,Qa~m,Qb~m)    
  #electron_donation_5_0:
  S(x~3) + PSII(P680~p,Qa~n,Qb~2m) -> S(x~0) + PSII(P680~n,Qa~n,Qb~2m)    
  #electron_donation_6_0:
  S(x~3) + PSII(P680~p,Qa~m,Qb~2m) -> S(x~0) + PSII(P680~n,Qa~m,Qb~2m)    
  #electron_transports_from_qa_to_qb_1:
  PSII(P680~p,Qa~m,Qb~n) <-> PSII(P680~p,Qa~n,Qb~m)   
  #electron_transports_from_qa_to_qb_2:
  PSII(P680~n,Qa~m,Qb~n) <-> PSII(P680~n,Qa~n,Qb~m)   
  #electron_transports_from_qa_to_reduced_qb_1:
  PSII(P680~p,Qa~m,Qb~m) <-> PSII(P680~p,Qa~n,Qb~2m)  
  #electron_transports_from_qa_to_reduced_qb_2:
  PSII(P680~n,Qa~m,Qb~m) <-> PSII(P680~n,Qa~n,Qb~2m)  
  #Qb_PQ_exchange_1:
  PQ() + PSII(P680~p,Qa~n,Qb~2m) <-> PQH() + PSII(P680~p,Qa~n,Qb~n)   
  #Qb_PQ_exchange_2:
  PQ() + PSII(P680~n,Qa~n,Qb~2m) <-> PQH() + PSII(P680~n,Qa~n,Qb~n)   
  #Qb_PQ_exchange_3:
  PQ() + PSII(P680~p,Qa~m,Qb~2m) <-> PQH() + PSII(P680~p,Qa~m,Qb~n)   
  #Qb_PQ_exchange_4:
  PQ() + PSII(P680~n,Qa~m,Qb~2m) <-> PQH() + PSII(P680~n,Qa~m,Qb~n)   
  #electron_donation_to_P700_by_Pc_1:
  Pc(x~n) + PSI(P700~p,Fb~n) <-> Pc(x~p) + PSI(P700~n,Fb~n)   
  #electron_donation_to_P700_by_Pc_2:
  Pc(x~n) + PSI(P700~p,Fb~m) <-> Pc(x~p) + PSI(P700~n,Fb~m)   
  #charge_separation_in_PSI:
  PSI(P700~n,Fb~n) -> PSI(P700~p,Fb~m)    
  #charge_separation_in_PSII_1:
  PSII(P680~n,Qa~n,Qb~n) -> PSII(P680~p,Qa~m,Qb~n)   
  #charge_separation_in_PSII_2:
  PSII(P680~n,Qa~n,Qb~m) -> PSII(P680~p,Qa~m,Qb~m)   
  #charge_separation_in_PSII_3:
  PSII(P680~n,Qa~n,Qb~2m) -> PSII(P680~p,Qa~m,Qb~2m) 
end reaction rules

\end{verbatim}
}


The model has many rules that are perfect candidates for reduction using \emph{context enumeration elimination}, such as
the following three \emph{charge separation} rules:
{\small\begin{verbatim}
PSII(P680~n,Qa~n,Qb~n) -> PSII(P680~p,Qa~m,Qb~n)  
PSII(P680~n,Qa~n,Qb~m) -> PSII(P680~p,Qa~m,Qb~m)   
PSII(P680~n,Qa~n,Qb~2m) -> PSII(P680~p,Qa~m,Qb~2m)
\end{verbatim}}
These rules can be reduced to the following single \emph{charge separation} rule:
{\small\begin{verbatim}
PSII(P680~n,Qa~n,Qb) -> PSII(P680~p,Qa~m,Qb)   
\end{verbatim}
}

After applying all \emph{context enumeration eliminations} we are able to reduce the model size from 69 rules down to just 22 rules.
Since the model is not using any binding sites, we can automatically apply \emph{generic unbound context eliminations} to further reduce the model.
The rule stated above is reduced to the following form:

{\small\begin{verbatim}
PSII(P680~n,Qa~n) -> PSII(P680~p,Qa~m)
\end{verbatim}
}

We managed to reduce the model significantly.
The original ideas of the author are much more obvious in this reduced model.
The reduced model is much easier to modify and extend than its original version.
 The resulting model is the following:\\
{\small\begin{verbatim}
begin reaction rules
  #electron_transport_from_Fd_to_FNR:  
    FNR(x~a) + Fd(x~m) <-> FNR(x~am) + Fd(x~n)   
    FNR(x~am) + Fd(x~m) <-> FNR(x~a2m) + Fd(x~n)  
  #activation_of_FNR:  
    FNR(x~i) -> FNR(x~a)
  #turnover_of_FNR:   
    FNR(x~a2m) -> FNR(x~a)
  #charge_separation_in_PSI:  
    PSI(P700~n,Fb~n) -> PSI(P700~p,Fb~m)  
  #electron_transport_from_photosystemI_to_Fd:
    Fd(x~n) + PSI(Fb~m) <-> Fd(x~m) + PSI(Fb~n)
  #electron_transports_from_qa_to_qb:      
    PSII(Qa~m,Qb~n) <-> PSII(Qa~n,Qb~m)
  #electron_transports_from_qa_to_reduced_qb:   
    PSII(Qa~m,Qb~m) <-> PSII(Qa~n,Qb~2m)
  #electron_transport_inside_cytochrome:    
    CytB6F(bL~m,bHc~n) <-> CytB6F(bL~n,bHc~m)
    CytB6F(bL~m,bHc~m) <-> CytB6F(bL~n,bHc~2m)
  #electron_transport_from_reduced_PQ_to_cytochrome:  
    PQH() + CytB6F(bL~n,f~n) <-> PQ() + CytB6F(bL~m,f~m) 
  #charge_separation_in_PSII:  
    PSII(P680~n,Qa~n) -> PSII(P680~p,Qa~m) 
  #electron_donation_to_P700_by_Pc: 
    Pc(x~n) + PSI(P700~p) <-> Pc(x~p) + PSI(P700~n)
  #electron_transport_from_cytochrome_to_oxidized_PQ:  
    PQ() + CytB6F(bHc~2m) <-> PQH() + CytB6F(bHc~n)  
  #electron_transport_from_reduced_Fd_to_cytochrome:     
    Fd(x~m) + CytB6F(bHc~n) <-> Fd(x~n) + CytB6F(bHc~m) 
    Fd(x~m) + CytB6F(bHc~m) <-> Fd(x~n) + CytB6F(bHc~2m)
  #electron_donation: 
    S(x~0) + PSII(P680~p) -> S(x~1) + PSII(P680~n)  
    S(x~1) + PSII(P680~p) -> S(x~2) + PSII(P680~n)  
    S(x~2) + PSII(P680~p) -> S(x~3) + PSII(P680~n)  
    S(x~3) + PSII(P680~p) -> S(x~0) + PSII(P680~n)  
  #Qb_PQ_exchange:         
    PQ() + PSII(Qb~2m) <-> PQH() + PSII(Qb~n)  
  #electron_transport_from_cytochrome_to_oxidized_Pc:
    CytB6F(f~m) + Pc(x~p) <-> CytB6F(f~n) + Pc(x~n) 
end reaction rules
\end{verbatim}
}

For the purpose of further analysis,the following observables computing the fluorescence emission are defined~\cite{lazar2009big}:
{\small\begin{verbatim}
begin observables
  Molecules Qa_m     PSII(Qa~m)
  Molecules PQ_n     PQ()
end observables
\end{verbatim}
}

The computed fluorescence signals are defined by the following assignements over the observables:
\begin{itemize}
    \item \emph{unquenched fluorescence}:
        \[F_{unq}(t) = \frac{0.45 Qa_m(t)}{1 - 0.55 Qa_m(t)}\]
    \item \emph{quenched fluorescence}:
        \[F_q(t) = \frac{F_{unq}(t)}{1 + ({1 \over 45} + {4 \over 63} Qa_m(t)) PQ_n(t)} \]
\end{itemize}


Figure~\ref{fig:lazar_rb_sim_ssa} shows signals obtained by averaging the results from 20 sto\-chastic simulations of the model.
We had to convert the deterministic rates to stochastic ones to obtain these results.
Details of this conversion can be seen in~\cite{thesis}.

\begin{figure}[ht]
\begin{center}
\includegraphics[width=.5\textwidth]{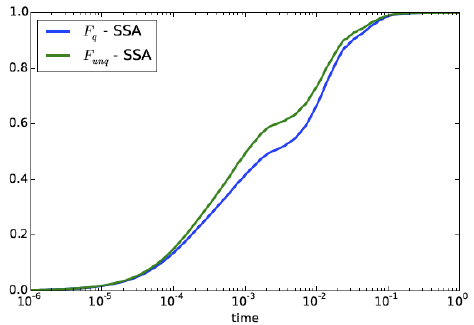}
\caption{ 
    Averaged signals computed from the outputs of 20 stochastic simulations of Laz\'{a}r's model. The curves show the time course of light emmission intensity in relative units.
 }
\label{fig:lazar_rb_sim_ssa}
\end{center}
\end{figure}

In the model presented above, a module of photosystem II including the oxygen evolving complex is considered. In~\cite{lazar2009approaches}, Laz\'{a}r ellaborates on several models of photosystem II with varying levels of complexity integrating the following components of PSII:



\begin{itemize}
\item Electron donor $P680$
\item $S_i$ states of the oxygen evolving complex (OEC) for $i = 0,1,2,3$.
\item Primary quinone electron acceptor $Q_a$
\item Secondary quinone electron acceptor $Q_b$
\item PQ (plastoquinone) pool
\end{itemize}

The measured output of the models is demonstrated in Figure~\ref{fig:lazar_ps2_sim} -- the fluorescence rise (FLR) in chlorophyll \emph{a}.

Laz\'{a}r considers several simplifications of the PSII module and describes their effect on the shape of the fluorescence signal:
\begin{enumerate}
\item The $S_i$-states of the OEC are considered as separated from PSII. 
    In BNGL (specifying only the relevant sites), instead of \spec{PSII(S\~0\~1\~2\~3)} we would have \spec{PSII()} and \spec{S(p\~0\~1\~2\~3)}.
    When describing the model using differential equations this simplification reduces the number of equations from 64 down to 20. 
    With this simplification the system turns into second order kinetics instead of first order kinetics. 
    Illustrated in BNGL: \\ \spec{PSII(P680\~p,S\~0) -> PSII(P680\~n,S\~1)} changes to\\ \spec{PSII(P680\~p) + S(p\~0) -> PSII(P680\~n) + S(p\~1)}

\item  Exchange of the double reduced $Q_b$ with a PQ molecule from the PQ pool is described by one second order reaction instead of two subsequent reactions. 
    In BNGL, instead of \\
\spec{PSII(Qb!1).PQ(Qb\~2m!1) <-> PSII(Qb) + PQ(Qb\~2m)} and \\
\spec{PSII(Qb) + PQ(Qb\~n) <-> PSII(Qb!1).PQ(Qb\~n!1)} \\
we have \\
\spec{PSII(Qb\~2m) + PQ(Qb\~n) <-> PSII(Qb\~n) + PQ(Qb\~2m)}. \\
This simplification makes the reading of the model a bit easier. 
By using this simplification along with the previous one the model is further reduced to 16 equations.

\item Only equations that follow a logical order are taken into account.
\end{enumerate}

These simplifications are combined into 8 models (two unique choices for each simplification -- as shown in Table~\ref{tb:models}). 
Different model formulations lead to different simulations of the fluorescence rise. Laz\'{a}r concludes that it is very important to care about how we formulate the model, because some simplifications may lead to completely unusual FLR behaviour. 
Models 7 and 8 resulted in a behaviour very similar to the O-J-I-P behaviour observed in \emph{in vivo} experiments.
We know that OEC is actually bound to PSII and that the process of reducing $PQ$ by $Q_b$ proceeds in two steps. 
These are two more arguments that proof the model 8 to be the most accurate representation of the real-world system.
As we present later, the size of model 8 is greatly reduced by rule-based modelling. 
Its size is even smaller than the size of the smallest model listed here.

\begin{table}[h]
\centering
{\small
\tabcolsep=0.2cm
\begin{tabular}{lcccccccc}
\toprule
Model no.:   &  1  &  2  &  3  &  4  &  5  &  6  &  7  &  8  \\
\midrule
OEC separated from PSII   & yes & yes & yes & yes &  no &  no &  no &  no \\
PQ exchange by 1 reaction & yes & yes &  no &  no & yes & yes &  no &  no \\
Only logical equations    &  no & yes &  no & yes & yes &  no & yes &  no \\
\bottomrule
\end{tabular}}
\caption{ Model variants of PSII considered by Laz\'{a}r.}
\label{tb:models}
\end{table}

To demonstrate the advantage of rule-based modelling to easily modify and compose the individual model variants, we take the most complex model of PSII Laz\'{a}r presented in~\cite{lazar2009approaches} (model 8) and we reformulate it in BNGL.   
Simulation of this model is consistent with O-J-I-P transient (see Figure~\ref{fig:lazar_ps2_sim}).
In its original form, this model consists of 64 differential equations that can be rewritten to following 10 rules:

{\small\begin{verbatim}
begin reaction rules
  # Light induced charge separation between P680 and Qa and charge recombination:
    PSII(P680~n,Qa~n) <-> PSII(P680~p,Qa~m)
  # Electron donation from S-states of OEC to P680+
    PSII(P680~p,S~0) -> PSII(P680~n,S~1)  
    PSII(P680~p,S~1) -> PSII(P680~n,S~2) 
    PSII(P680~p,S~2) -> PSII(P680~n,S~3) 
    PSII(P680~p,S~3) -> PSII(P680~n,S~0)
  # Electron transport from Qa- to Qb:
    PSII(Qa~m,Qb~red!1).PQ(Qb~n!1) <-> PSII(Qa~n,Qb~red!1).PQ(Qb~m!1)  
    PSII(Qa~m,Qb~red!1).PQ(Qb~m!1) <-> PSII(Qa~n,Qb~red!1).PQ(Qb~2m!1) 
  # Exchange of doubly reduced Qb with oxidized PQ molecule 
  # from the PQ pool by two subsequent reversible reactions:
    PSII(Qb~red!1).PQ(Qb~2m!1) <-> 
      PSII(Qb~red) + PQ(Qb~2m) 
    PSII(Qb~red) + PQ(Qb~n) <-> 
      PSII(Qb~red!1).PQ(Qb~n!1)  
  # Reversible reoxidation of reduced PQ molecules from PQ pool:
    PQ(Qb~2m) <-> PQ(Qb~n)  
end reaction rules
\end{verbatim}
}

\begin{figure}[ht]
\begin{center}
\includegraphics[width=.6\textwidth]{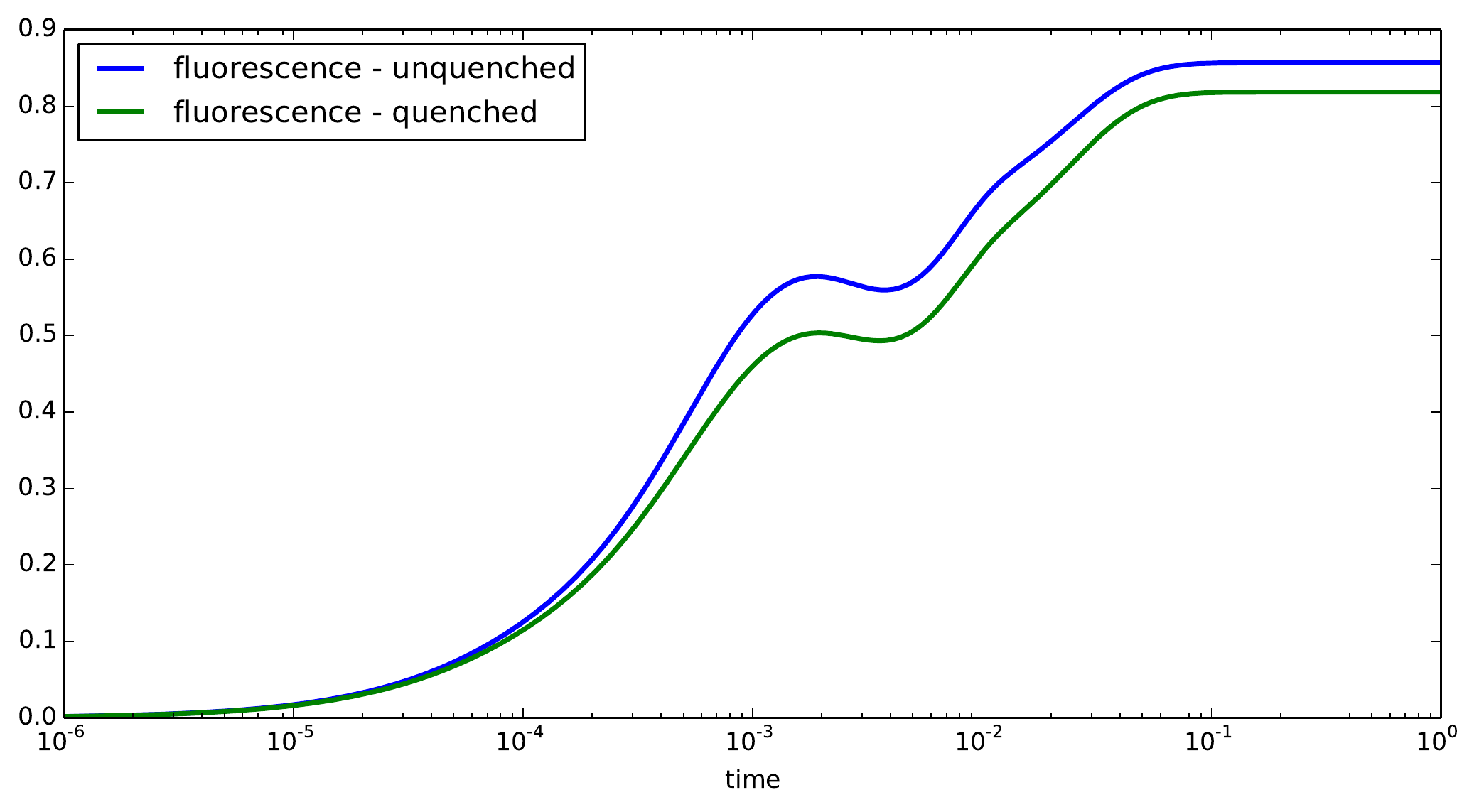}
\caption{ The fluorescence curves from the model of photosystem II. The curves show the time course of light emmission intensity in relative units.
 }
\label{fig:lazar_ps2_sim}
\end{center}
\end{figure}

Next, we replace the original photosystem II module in the integrated model with model 8.
This modification would be very hardly realisable when using the ODE approach to express the models.
The resulting model is listed below.

{\small\begin{verbatim}
begin molecule types
  PSII(P680~n~p,Qa~n~m,S~0~1~2~3,Qb~n~m)
  PQ(Qb~n~m~2m)
  CytB6F(bL~n~m,bHc~n~m~2m,f~n~m)
  Fd(x~n~m)
  Pc(x~n~p)
  PSI(P700~n~p,Fb~n~m)
  FNR(x~i~a~am~a2m)
end molecule types

begin reaction rules
  # Light induced charge separation/recombination between P680 and Qa:
    PSII(P680~n,Qa~n) <-> PSII(P680~p,Qa~m) 
  # Electron donation from S-states of OEC to P680+
    PSII(P680~p,S~0) -> PSII(P680~n,S~1)  
    PSII(P680~p,S~1) -> PSII(P680~n,S~2)  
    PSII(P680~p,S~2) -> PSII(P680~n,S~3)  
    PSII(P680~p,S~3) -> PSII(P680~n,S~0)  
  # Electron transport from Qa- to Qb:
    PSII(Qa~m,Qb~m!1).PQ(Qb~n!1) <-> PSII(Qa~n,Qb~m!1).PQ(Qb~m!1)   
    PSII(Qa~m,Qb~m!1).PQ(Qb~m!1) <-> PSII(Qa~n,Qb~m!1).PQ(Qb~2m!1)  
  # Exchange of doubly reduced Qb with oxidized PQ molecule 
  # from the PQ pool by two subsequent reversible reactions:
    PSII(Qb~m!1).PQ(Qb~2m!1) <-> PSII(Qb~m) + PQ(Qb~2m) 
    PSII(Qb~m) + PQ(Qb~n) <-> PSII(Qb~m!1).PQ(Qb~n!1)   
  # Electron transport from Fd to FNR:
    FNR(x~a) + Fd(x~m) <-> FNR(x~am) + Fd(x~n)  
    FNR(x~am) + Fd(x~m) <-> FNR(x~a2m) + Fd(x~n)
  # Activation of FNR:
      FNR(x~i) -> FNR(x~a) 
  # Turnover of FNR:
    FNR(x~a2m) -> FNR(x~a) 
  # Charge separation in PSI:  
    PSI(P700~n,Fb~n) -> PSI(P700~p,Fb~m) 
  # Electron transport from PSI to Fd:     
    Fd(x~n) + PSI(Fb~m) <-> Fd(x~m) + PSI(Fb~n) 
  # Electron transport inside cytochrome:    
    CytB6F(bL~m,bHc~m) <-> CytB6F(bL~n,bHc~2m) 
    CytB6F(bL~m,bHc~n) <-> CytB6F(bL~n,bHc~m) 
  # Electron transport from reduced PQ to cytochrome:  
    PQ(Qb~2m) + CytB6F(bL~n,f~n) <-> PQ(Qb~n) + CytB6F(bL~m,f~m)
  # Electron donation to P700 by Pc: 
    Pc(x~n) + PSI(P700~p) <-> Pc(x~p) + PSI(P700~n)
  # Electron transport from cytochrome to oxidized PQ:  
    PQ(Qb~n) + CytB6F(bHc~2m) <-> PQ(Qb~2m) + CytB6F(bHc~n)
  # Electron transport from reduced Fd to cytochrome:     
    Fd(x~m) + CytB6F(bHc~n) <-> Fd(x~n) + CytB6F(bHc~m)
    Fd(x~m) + CytB6F(bHc~m) <-> Fd(x~n) + CytB6F(bHc~2m)
  # Electron transport from cytochrome to oxidized Pc:
    CytB6F(f~m) + Pc(x~p) <-> CytB6F(f~n) + Pc(x~n)
end reaction rules
\end{verbatim}
}

In Figure~\ref{fig:ex_flu}, the outputs of the orginal and the combined model are compared.

\begin{figure}[ht]
\begin{center}
\includegraphics[width=\textwidth]{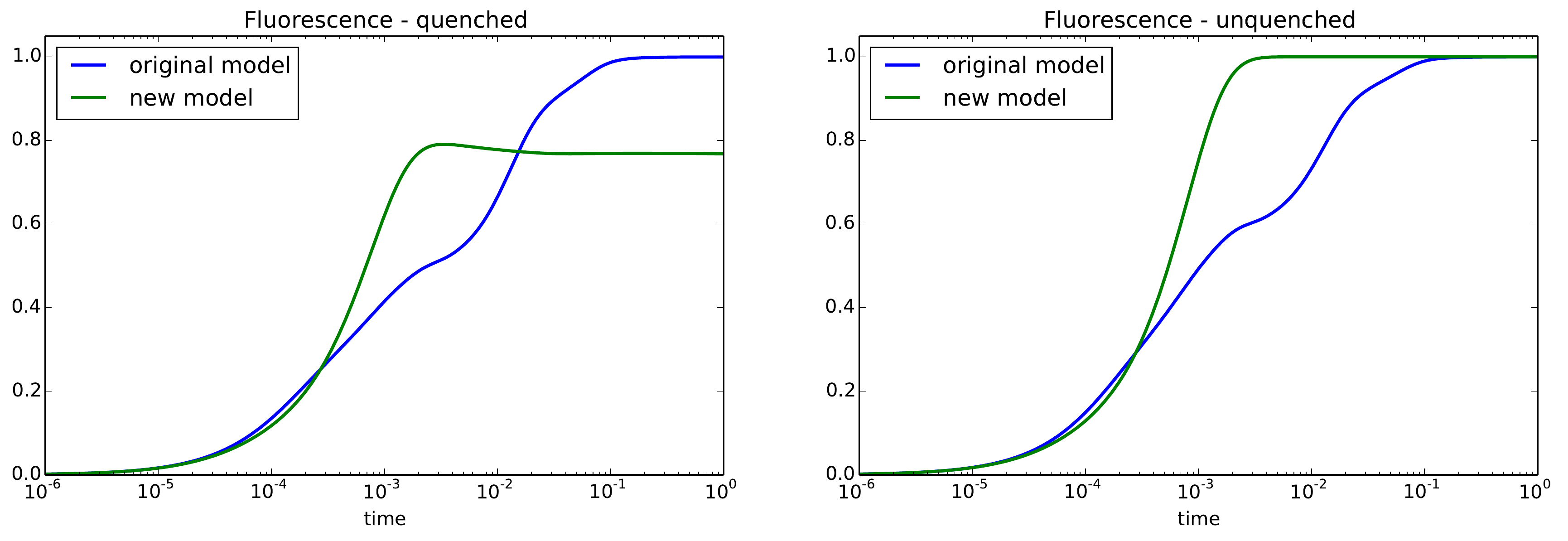}
\caption{ 
    Comparison of the fluorescence curves of the original and the combined model.
 }
\label{fig:ex_flu}
\end{center}
\end{figure}

\section{Conclusions}

We have demonstrated the unsuitability of traditional reaction-based modelling approaches for modelling complex biochemical processes, such as photosynthesis.
We explored existing models of photosynthesis and described the simplifications that were made in those models in order to battle the problem of combinatorial explosion.
We showed how these simplifications are undesirable.
Rule-based modelling allows us to compactly model the processes of photosynthesis in their full mechanistic complexity without the need for such simplifying assumptions.

We set on to naively reformulate selected representative models of photosynthesis as rule-based models.
These reformulated models were unnecessarily large, not exploiting the advantages of the rule-based format.
Therefore, we formally defined several intuitive syntactic operations that can be used to reduce the size of these models.
We provided a case study where we implemented these operations so they can be performed automatically and we managed to achieve large reductions in the size of the models.
The order in which we applied the reductions turned out to be satisfactory.

We believe that in the future, the communities of biologists who are modelling photosynthesis consider the use of rule-based modelling.
Rule-based modelling brings in many advantages and eliminates the reason of some artificial model simplifications.

\bibliographystyle{eptcs}
\bibliography{papers}

\end{document}